\documentclass[letterpaper]{sig-alternate}

\usepackage{caption}
\usepackage[subrefformat=parens,labelformat=parens,justification=centering]{subfig}
\usepackage{relsize}
\usepackage{array}
\usepackage{multirow}
\usepackage{amsmath}
\usepackage{amssymb}
\usepackage{amsfonts}
\usepackage{mathtools}
\usepackage{times}
\usepackage{datetime}
\usepackage{graphicx}
\usepackage{enumerate}
\usepackage[mathscr]{euscript}
\usepackage{cases}
\usepackage{ifpdf}
\usepackage{epstopdf}

\usepackage{etoolbox}
\makeatletter
\patchcmd{\maketitle}{\@copyrightspace}{}{}{}
\makeatother

\ifpdf \fi

\newcommand{\todo}[1]{{\bf \textcolor{blue}{[*** #1 ***]}}}
\newcommand{\note}[1]{{\bf \textcolor{red}{#1}}}
\newcommand{\comments}[1]{}

\newcommand{\floor}[1]{\left\lfloor #1 \right\rfloor}
\newcommand{\ceil}[1]{\left\lceil #1 \right\rceil}
\newcommand{\abs}[1]{\left\lvert #1 \right\rvert}

\renewcommand{\vec}[1]{{\bf #1}}
\newcommand{\hpref}{\succ_h}
\newcommand{\hprefeq}{\succeq_h}

\let\oldchi\chi
\renewcommand{\chi}{{\raisebox{2pt}{$\oldchi$}}}  
\newcommand{\Haug}{\widetilde{H}}  

\DeclareMathOperator{\cit}{cit}  
\DeclareMathOperator{\Util}{Util}  

\newtheorem{theorem}{Theorem}
\newtheorem{lemma}[theorem]{Lemma}

\newtheorem{proposition}[theorem]{Proposition}

\newenvironment{itemize*}
	{\vspace{-6pt} \begin{itemize} \setlength{\itemsep}{1pt} \setlength{\parskip}{0pt} \setlength{\parsep}{0pt}}
	{\end{itemize} \vspace{-4pt}}
\newenvironment{enumerate*}[1][]
	{\vspace{-6pt} \begin{enumerate}[#1] \setlength{\itemsep}{1pt} \setlength{\parskip}{0pt} \setlength{\parsep}{0pt}}
	{\end{enumerate} \vspace{-4pt}}

\begin{document}
%
\conferenceinfo{WWW Workshop on Big Scholarly Data}{'14 Seoul, Korea}

\title{Modeling Collaboration in Academia: \\ A Game Theoretic Approach}

\numberofauthors{2}
\author{
\alignauthor Qiang Ma, S. Muthukrishnan, Brian Thompson\\
       \affaddr{Rutgers University}\\
       \affaddr{Dept. of Computer Science}\\
       \affaddr{Piscataway, NJ 08854, USA}\\
       \email{\{qma,muthu\}@cs.rutgers.edu, bthompso8784@gmail.com}
\alignauthor Graham Cormode\\
       \affaddr{University of Warwick}\\
       \affaddr{Dept. of Computer Science}\\
       \affaddr{Coventry CV4 7AL, UK}\\
       \email{g.cormode@warwick.ac.uk}
}

\date{}

\maketitle

\begin{abstract}

In this work, we aim to understand the mechanisms driving academic collaboration. We begin by building a model for how researchers split their effort between multiple papers, and how collaboration affects the number of citations a paper receives, supported by observations from a large real-world publication and citation dataset, which we call the {\em h-Reinvestment model}. Using tools from the field of Game Theory, we study
researchers' collaborative behavior over time under this model, with the premise that each researcher wants to maximize his or her academic success. We find analytically that there is a strong incentive to collaborate rather than work in isolation, and that studying collaborative behavior through a game-theoretic lens is a promising approach to help us better understand the nature and dynamics of academic collaboration.

\end{abstract}

\section{Introduction}
\label{sec:intro}

Researchers exhibit a wide range of work habits and behaviors. Some work on many papers simultaneously, while others focus on only a few at a time. Some engage in mentoring relationships, while others choose to collaborate mostly with their peers, and still others prefer to work independently. These behaviors may be motivated by a variety of factors such as institutional needs, academic field, stage in career, funding situation, and affinity for teaching. We pose the question: ``If researchers were motivated by $X$, what would the world of academic research look like?'' In the current work, we analyze collaborative behavior in a large scholarly dataset and arrive at a model of academic collaboration, which we call the {\em h-Reinvestment model}. We then formalize a game based on this model, and study the outcome of the game when each researcher tries to optimize a given objective function.

Our first result is that two researchers perform asymptotically better by collaborating than by publishing only independent work. That is,
collaboration is preferable to isolation (given our assumptions). Our second result is that when researchers are constrained to following the same strategy every year, the best outcome is when they arrange themselves into (stable) pairs to work together. Yet when the researchers are allowed to change their strategies over time, this scheme no longer represents a stable equilibrium. This highlights an important limitation of the existing literature, most of which is based on models where collaboration strategies remain constant over time. Our model and approach provide a new framework for further study, which can help us to better understand the dynamics of collaborative systems. Our main contributions are:
\begin{trivlist}
\item $\bullet$~The {\em h-Reinvestment model}, an academic collaboration model supported by detailed analysis of publication data
\item $\bullet$~The {\em Academic Collaboration game}, where researchers collaborate to maximize their academic success
\item $\bullet$~An analysis of collaboration strategies and game equilibria
\end{trivlist}

\subsection{Related Work}
\label{sec:related}

Some existing work studies the academic success of a researcher over time. While introducing the h-index, Hirsch suggested a model in which a researcher publishes a constant number of papers per year, and each paper receives a constant number of additional citations per year, which results in linear growth of the h-index~\cite{Hirsch@PNAS05}.
Guns and Rousseau suggested a peak-decay citation model, and showed through simulations that by varying the parameters -- or choosing them stochastically -- growth of the h-index can be linear, concave, or S-shaped~\cite{GR@JASIS09}. Kleinberg and Oren proposed a game-theoretic model to analyze how researchers choose which open problems to work on, and how credit gets attributed when multiple researchers solve the same problem~\cite{KO@STOC11}. None of this previous work models collaboration explicitly.

There has been effort in recent years to design bibliographic metrics that take collaboration into account \cite{AAH@Scientometrics10,KPMVD@CIKM10,CMMT@JOI13}.
All of these approaches model the academic environment as a static graph representing co-author relationships. However, in real life, a researcher's behavioral patterns may change over time. We suggest that more sophisticated models are needed to understand the intricacies of collaborative behavior. Cardillo et al.\ empirically study the correlation between stability of local graph structure over time and the willingness of individuals to compromise their own interests in favor of social cooperation, but stop short of suggesting a mechanism that would explain such behavior~\cite{CPNSGL@arXiv13}.

\section{Methodology}
\label{sec:method}

We introduce some  terminology and notation, summarized in 
Table~\ref{table:notation}.
Using a game-theoretic framework, we then describe a game of academic collaboration with which we can simulate researchers' collaborative behavior over time.

\subsection{Preliminaries}
\label{sec:collab-prelim}

\begin{table}[!t]
	\centering
	\begin{tabular}{|c|l|}
	\hline
	\textbf{Notation} & \textbf{Description} \\ \hline
	$\cit(p)$ & the total \# of citations received by paper $p$ \\ \hline
	$A(p)$ & the set of authors of paper $p$ \\ \hline
	$P(a)$ & the set of papers authored by $a$ \\ \hline
	$P_y(a)$ & the set of papers authored by $a$ in year $y$ \\ \hline
	$\chi_y(a)$ & the citation profile of researcher $a$ in year $y$ \\ \hline
	$h_y(a)$ & the h-index of researcher $a$ in year $y$ \\ \hline
	$H_y(a)$ & the h-profile of researcher $a$ in year $y$ \\ \hline
	$\Haug_y(a)$ & the h-augmenting profile of researcher $a$ in year $y$ \\ \hline
	\end{tabular}
	\caption{Table of basic notation}
	\label{table:notation}
\end{table}

We define the {\em citation profile} of a set of papers $P$, denoted
$\chi(P)$, to be the multi-set $\{\cit(p) : p \in P\}$; and the
citation profile of a researcher $a$ to be $\chi(a) = \chi(P(a))$. 
Then $\chi_y(a)$ denotes the citation profile of researcher $a$ in year $y$.

We define the {\em h-index} (generalizing Hirsch~\cite{Hirsch@PNAS05}) of a multi-set of non-negative integers $Z$, denoted $h(Z)$, to be the largest integer $h$ such that at least $h$ elements of $Z$ are greater than or equal to $h$:
$$h(Z) = \max \left\{ h : \abs{\{z \in Z,\ z \geq h\}} \geq h \right\}.$$
For simplicity of notation, 
we let $h(a) = h(P(a)) = h(\chi(P(a)))$, and let $h_y(a)$ denote the h-index of researcher $a$ in year $y$.

We define the {\em h-profile} of a multi-set of non-negative integers $Z$, denoted $H(Z)$, to be the multi-set of integers in $Z$ that are greater than or equal to $h(Z)$:
$$H(P) = \{z \in Z : z \geq h(Z)\}.$$
We similarly define
$H(a) = H(P(a)) = H(\chi(P(a)))$, and let $H_y(a)$ denote the h-profile of researcher $a$ in year $y$.

Sometimes we are only interested in the papers with strictly more than $h$ citations. We define the {\em h-augmenting profile} of a multi-set of non-negative integers $Z$, denoted $\Haug(Z)$, to be the multi-set of integers in $Z$ that are strictly greater than $h(Z)$:
$$\Haug(P) = \{z \in Z : z > h(Z)\}.$$
We similarly define
$\Haug(a) = \Haug(P(a)) = \Haug(\chi(P(a)))$, and let $\Haug_y(a)$ denote the h-augmenting profile of researcher $a$ in year $y$. Intuitively, the h-augmenting profile indicates progress towards increasing the h-index.

Let $Z$ and $Z'$ be multi-sets of non-negative integers. We say $Z$ is {\em weakly h-preferable} to $Z'$, denoted $Z \hprefeq Z'$, if $h(Z) \geq h(Z')$ and $(\forall\  z_0 > h(Z))\ \abs{\{z \in Z : z \geq z_0\}} \geq \abs{\{z \in Z' : z \geq z_0\}}$. We say $Z$ is {\em strongly h-preferable} to $Z'$, denoted $Z \hpref Z'$, if in addition either $h(Z) > h(Z')$ or $\exists\ z_0 > h(Z)$ for which the inequality is strict. When $P$ and $P'$ are two sets of papers, we write $P \hprefeq P'$ to denote that $\chi(P) \hprefeq \chi(P')$, and $P \hpref P'$ to denote that $\chi(P) \hpref \chi(P')$.

Next, we propose a model of academic collaboration.

\subsection{The h-Reinvestment Model}
\label{sec:collab-model}

We investigate how researchers distribute their effort between multiple papers, and the correspondence between a paper's success and the effort invested in it by its coauthors, by analyzing a large corpus of Computer Science publications.
We extract all publications, along with authors and number of citations received, from a snapshot of the DBLP database, which contains approximately 1 million researchers and 2 million publications. The size and variety of this data mean that it is possible to validate
and calibrate our model from this dataset.
Empirical observations lead us to a model where in each year $y$, each researcher $a$ has some amount of {\em research potential} $Q_y(a)$ to be invested in writing papers, which is proportional to his or her academic success up to that point.

We first analyze the simple case of a paper published by a single author who had no other publications that same year,\footnote{The assumption is that if a researcher published only one paper in a given year, then all of her effort went into that paper. In reality, she could have worked on projects that were not published that year, but that is hard to evaluate empirically since unpublished papers are not captured in the data.} and explore the relationship between the number of citations a paper receives and several attributes of the author: number of papers published previously, total number of citations received previously, and current h-index. We compute Spearman's rank correlation coefficient for each of the attributes,\footnote{We choose this over the more common Pearson's coefficient because it is more robust to non-linear relationships.} and find that the h-index has the highest correlation with a value of 0.34, compared to paper count with a value of 0.28 and citation sum with a value of 0.08. Therefore, in subsequent analysis, we use the h-index as a proxy for the research potential of an author.

In Figure~\ref{fig:model-single-author}, we take a closer look at the relationship between the h-index of the author and the number of citations a paper receives. The plot shows the median number of citations received on papers singly-authored by a researcher with h-index $h$ for each value of $h$. We use the median because there are a few extreme outliers which skew the average to the right, and we are looking for a model which represents a typical researcher. Comparison to the best-fit line demonstrates visually that the two quantities have a linear relationship up until an h-index of about 10, indicating that the number of citations a single-authored paper receives is proportional to the h-index of the author when he puts all of his effort into the paper. For values of $h > 10$, the fluctuation may be a result of high variance and too few data points.

\begin{figure}[t]
\centering
		\includegraphics[width=0.92\columnwidth]{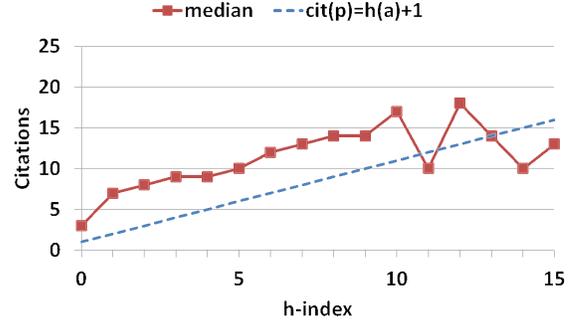}
\caption{The h-index of the author versus the median number of citations received across all single-authored papers for which the author published no other papers the same year.}
	\label{fig:model-single-author}
\end{figure}

Next, we look at the case of papers with multiple authors. To isolate this aspect of the model, we consider two-author papers where neither of the authors published any other papers in the same year.
In Figure~\ref{fig:model-two-authors},
we plot the sum of the h-indices of the authors versus the median number of citations received across all such papers. We again observe
a linear relationship, supporting that the combined research potential of multiple authors is {\em additive} when they put all of their effort into the paper.

\begin{figure}[t]
	\centering
		\includegraphics[width=0.92\columnwidth]{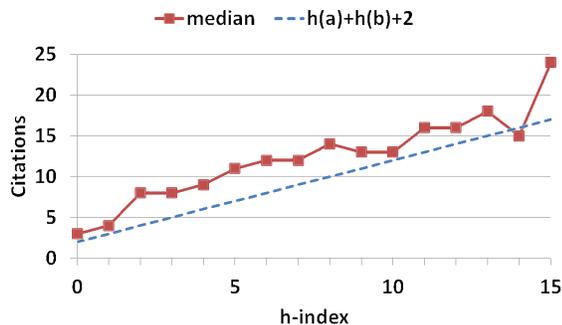}
		\label{fig:plot-two-authors}
	\caption{The sum of the h-indices of the coauthors versus the median number of citations received across all two-authored papers where neither author published any other papers in the same year. The dashed line is the number of citations predicted by our model.}
	\label{fig:model-two-authors}
\end{figure}

Finally, we investigate what happens when an author publishes multiple papers in the same year by narrowing our focus to instances where
aside from the author of interest, none of the coauthors published any other papers in the same year. 
Using the previous result of research potential being additive across multiple coauthors, we plot the h-index of the author of interest against the value
$$\textstyle{\sum_{p \in P_y(a)} \Big( \cit(p) - \sum_{b \in A(p) \backslash \{a\}} h(b) \Big)}$$
in Figure~\ref{fig:model-one-author-two-papers}.
The plot shows a linear relationship, suggesting that the allocation of an author's research potential across multiple papers is also linear.

\begin{figure}[t]
	\centering
	\includegraphics[width=0.92\columnwidth]{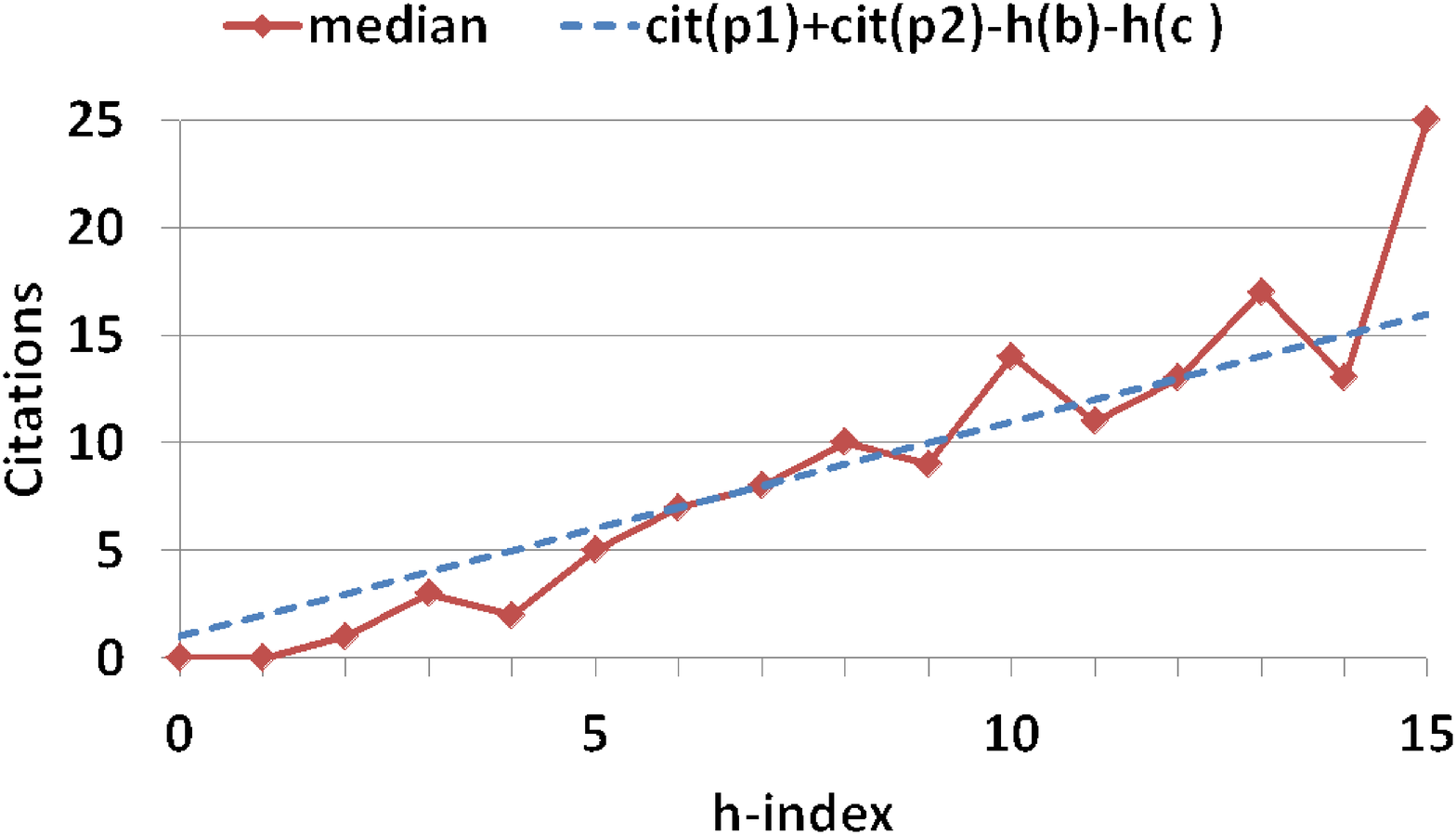}
		\label{fig:plot-one-author-two-papers}
	\caption{The h-index of an author $a$ who published two papers $p_1$ and $p_2$ in the same year with coauthors $b_1$ and $b_2$ respectively, neither of which published any other papers in the same year, versus the median of $(\cit(p_1) - h(b_1)) + (\cit(p_2) - h(b_2))$.}
	\label{fig:model-one-author-two-papers}
\end{figure}

We formalize the \textit{h-Reinvestment model} based on the observations above:
\begin{enumerate}
\addtolength{\itemsep}{-1ex}
\item In year $y$, a researcher $a$ has $Q_y(a) = h_y(a) + 1$ units of research potential to be invested in writing papers.\footnote{The purpose of the $+1$ is to signify that an author working independently will continue to make progress.}
\item Each researcher distributes her research potential between some number of papers to be published that year.
\item A paper $p$ will receive citations equal to the sum of the research potential invested by its coauthors.
\end{enumerate}

\subsection{The Academic Collaboration Game}
\label{sec:collab-game}
We appeal to the field of Game Theory, and define a game based on the h-Reinvestment model in Section~\ref{sec:collab-model}.
For simplicity of analysis, we model all citations as being received in the same year that the paper is published. We also assume that there is a practical limit on the number of coauthors that can meaningfully contribute to a paper, and in the following analysis limit a paper to two coauthors. Future work is to revisit the analysis under more general models.

A {\em game} is a way of modeling the decisions of a set of rational {\em players} whose {\em actions} collectively determine an {\em outcome}. A player's goal is to achieve an outcome of maximal {\em utility} to that player. We model collaboration in academia as a {\em repeated game}, where the same base game is played multiple times, and in each iteration players choose actions simultaneously.

We formalize a repeated game played by a set of researchers, explicitly defining the actions available to each researcher in each year, the outcomes determined by those actions, and the utility of each possible outcome to each researcher. We refer to this as the {\em Academic Collaboration (AC) game}:
\begin{trivlist}
\item 
$\bullet$~{\bf Players:} Let $A$ be a set of researchers, each $a \in A$ initially having published a set of papers resulting in citation profile $\chi_0(a)$.
\item 
$\bullet$~{\bf Actions:} In year $y$, each researcher $a \in A$ has $Q_y(a)$ units of research potential to allocate among individual and collaborative papers. Formally, $a$ constructs a finite sequence of non-negative integers $\vec{q_y^a}$, and for each potential coauthor $a' \in A$ a sequence $\vec{q_y^{a,a'}}$, such that
$$\sum_i \vec{q_y^a}[i] + \sum_{a'} \sum_i \vec{q_y^{a,a'}}[i] = Q_y(a).$$
\item 
$\bullet$~{\bf Outcome:}
For each $\vec{q_y^a}[i] > 0$ and each $\vec{q_y^{a,a'}}[i] +
\vec{q_y^{a',a}}[i] > 0$, a paper is produced that receives citations
equal to the total of the research potential invested by its coauthors. A researcher $a$ becomes a coauthor on a paper $p$ by investing a non-zero amount of research potential in it.
\item 
$\bullet$~{\bf Utility:} The function $\Util_y(a) = h_y(a)$ indicates the utility for researcher $a$ at the end of year $y$.
\end{trivlist}

We will consider the AC game of {\em infinite horizon}, which means that each player wants to maximize his utility in the limit, rather than after some pre-specified number of years.\footnote{Although in reality a researcher lives for only a finite number of years, infinite games are arguably a reasonable model of human behavior when ``players examine a long-term situation without assigning a specific status to the end of the world''~\cite{Rubinstein@AET92}.} The Game Theory literature considers several ways to compare player preferences in infinite games. Our approach is most similar to the overtaking criterion presented in~\cite{Rubinstein@JET79}.

The {\em game state} represents, at any point in the game, all information that may help determine the available actions, corresponding outcomes, and utilities of the players. In the AC game, we define the game state to consist of the citation profiles of the researchers.

A {\em strategy} is a set of rules that govern which action a player will take given her knowledge of the game state. In the current work, we only consider deterministic strategies.

Let $s$ be a set of strategies for a game, one per player; this is referred to as a {\em strategy profile}. For the purpose of analysis, we take two strategy profiles to be equivalent if they always result in the same outcome. When considering multiple strategy profiles, we denote by $P^s_y(a)$, $\chi^s_y(a)$, $h^s_y(a)$, $H^s_y(a)$, $\Haug^s_y(a)$, and $\Util^s_y(a)$ the papers, citation profile, h-index, h-profile, h-augmenting profile, and utility, respectively, for player $a$ after $y$ iterations of the game when the players follow their respective strategies in $s$; and by $W^s(A)$ the social welfare under $s$. We denote by $s_a$ the strategy for player $a \in A$ under strategy profile $s$,
and by $s_{\bar{a}}$ the strategies for all players other than $a$; by $s_{A'}$ the strategies for players in $A' \subseteq A$, and by $s_{\bar{A'}}$ the strategies for players not in $A'$.

Let $f_n$ and $g_n$ be two infinite real-valued sequences. We say that $f_n$ {\em overtakes} $g_n$ if $\limsup\limits_{n \to \infty} f_n - g_n > 0$ and $\liminf\limits_{n \to \infty} f_n - g_n \geq 0$.\footnote{In~\cite{Rubinstein@JET79}, $f_n$ overtakes $g_n$ if $\liminf\limits_{n \to \infty} f_n - g_n > 0$. Our definition additionally allows for the situation in Figure~\ref{fig:overtaking-sequences-yes}.} We note that there are three (mutually exclusive and exhaustive) possibilities, illustrated in Figures~\ref{fig:overtaking-sequences-yes} and~\ref{fig:overtaking-sequences-no}:
\begin{itemize*}
\item $f_n$ overtakes $g_n$
\item $g_n$ overtakes $f_n$
\item neither $f_n$ nor $g_n$ overtakes the other
\end{itemize*}


\begin{figure}[t]
	\centering
	\includegraphics[width=0.92\columnwidth]{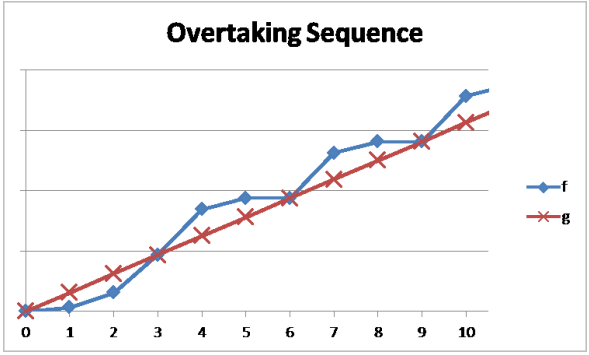}
	\caption{Example of overtaking sequences: $f_n$ overtakes $g_n$.}
	\label{fig:overtaking-sequences-yes}
\end{figure}
	
\begin{figure}[t]
	\centering
	\includegraphics[width=0.92\columnwidth]{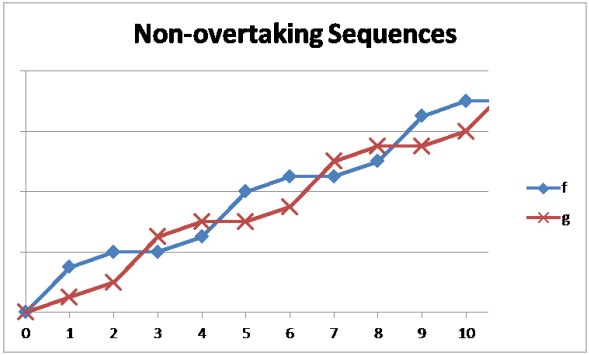}
	\caption{Example of non-overtaking sequences: Neither sequence $f_n$ nor $g_n$ overtakes the other.}
	\label{fig:overtaking-sequences-no}
\end{figure}



Multiple notions of equilibrium have been proposed in the literature. Due to the collaborative nature of the AC game, we consider a set of strategies to be in equilibrium if no two researchers would prefer to deviate from their current strategies in order to collaborate with one another. We formalize this by generalizing the notion of stability presented in~\cite{GS@AMM62}.

Given a strategy profile $s$ for the players in an infinite game, we say that the subset of players $A' \subseteq A$ is {\em unstable under $s$} if there exist alternate strategies $s'_{A'}$ for the players in $A'$ such that $(\forall\ a \in A')\ \Util^{s_{\bar{A'}} \cup
  s'_{A'}}_n(a)$ overtakes $\Util^s_n(a)$. We define a strategy profile $s^\ast$ to be a {\em $k$-stable equilibrium} if there does not exist an unstable set of size at most $k$. Subsequently, we use the term {\em equilibrium} to refer to a $2$-stable equilibrium.

\section{Evaluation}
\label{sec:eval}

\noindent
We use the model of the previous section to study the AC game.

\subsection{Single-Player Game}
\label{sec:collab-1pgame}

First, we consider the AC game when there is only one player, researcher $a$. In this case, $a$ may only write single-author papers; the question is how many papers to write and how to optimally distribute her research potential among them.

We begin by analyzing how the utility function grows when $a$ puts all of her effort into writing a single paper each year.

\begin{proposition}
\label{thm:collab-1p-asymptotics}
	Consider the single-player AC game of infinite horizon. Let $s^\ast$ denote the strategy of investing all research potential into a single paper each year. Then the limit behavior for $a$'s utility under $s^\ast$ is
{$\limsup_{n \to \infty}\ {Util}^{s^\ast}_n(a) \sim \sqrt{2n}.$}
\end{proposition}

\begin{proof}
If the claim holds for $h_0(a) = 0$, then it also holds for $h_0(a) > 0$, so assume that $h_0(a) = 0$. Following strategy $s^\ast$, from the time $a$ reaches an h-index of $h'$, it will take $h'+1$ years to accumulate $h'+1$ papers with $h'+1$ citations each. Thus $a$ requires a total of $n = \sum_{i=1}^h i = \frac{h(h+1)}{2}$ years to achieve an h-index of $h$. Conversely, as the number of years $n$ goes to infinity, $a$ achieves a utility of

\hspace*{0.2in}$
	\limsup\limits_{n \to \infty}\ {Util}^{s^\ast}_n(a) =
	\limsup\limits_{n \to \infty}\ \textstyle{\floor{\frac{-1 + \sqrt{1 + 8n}}{2}}}
  \sim \sqrt{2n}. 
$\hspace*{0.1in}
\end{proof}

We now compare $a$'s success under the single-paper strategy relative to other possible ways of distributing her effort.

\begin{lemma}
\label{thm:collab-1p-overtake}
	Consider the single-player AC game of infinite horizon. Let $s^\ast$ denote the strategy profile where each year the player $a$ invests all research potential into a single paper. Then for all strategy profiles $s \neq s^\ast$, $h^{s^\ast}_n(a)$ overtakes $h^{s}_n(a)$.
\end{lemma}

\begin{proof}
Consider a strategy profile $s \neq s^\ast$. Let $y^\ast$ be the first
year in which the outcome is different under $s^\ast$ and $s$, so that 
\[h_{y^\ast-1}(a) = h^{s^\ast}_{y^\ast-1}(a) = h^{s}_{y^\ast-1}(a)
\quad \text{ and }\]
\[\Haug_{y^\ast-1}(a) = \Haug^{s^\ast}_{y^\ast-1}(a) =
\Haug^{s}_{y^\ast-1}(a).\]
 Since $s^\ast$ produces a single paper that will receive $Q_{y^\ast}(a) = h_{y^\ast-1}(a) + 1$ citations, $a$'s strategy under $s$ must split the research potential between at least two papers, each thus receiving at most $h_{y^\ast-1}(a)$ citations, resulting in $H^{s^\ast}_{y^\ast}(a) \hpref H^{s}_{y^\ast}(a)$. It follows by induction that $H^{s^\ast}_y(a) \hpref H^{s}_y(a)$ for all $y \geq y^\ast$, and furthermore, that $h^{s^\ast}_y(a) > h^{s}_y(a)$ for all years $y \geq y^\ast$ in which $h^{s^\ast}(a)$ increases. By definition, $h^{s^\ast}_n(a)$ overtakes $h^{s}_n(a)$.
\end{proof}


%

\begin{theorem}
\label{thm:collab-1p-equilibrium}
	Consider the single-player AC game of infinite horizon. Let $s^\ast$ denote the strategy profile where each year the player $a$ invests all research potential into a single paper. Then $s^\ast$ is the only equilibrium.
\end{theorem}

\begin{proof}
This follows directly from Lemma~\ref{thm:collab-1p-overtake}.
\end{proof}

%
%

We have shown the strategy described above to be optimal for the single-player AC game. However, a researcher may hope to have a greater impact by collaborating with others. We explore this possibility in the following sections.

\subsection{Two-Player Game}
\label{sec:collab-2pgame}

We now consider the AC game with two players, $a$ and $a'$. For simplicity, we only analyze the case where $H_0(a) = H_0(a')$, i.e. initially both researchers have the same h-profile; the results can be generalized for arbitrary initial citation profiles. Note that if all papers produced through year $y$ are joint between $a$ and $a'$, then $h_y(a) = h_y(a')$, $H_y(a) = H_y(a')$, and $\Haug_y(a) = \Haug_y(a')$, in which case we will denote them by $h_y$, $H_y$, and $\Haug_y$, respectively.

We begin by considering two collaborative strategy profiles, and
analyze how the players' utility grows under in each case: one where
both players pool all their effort into a single joint paper, and
another where they collaborate on two papers simultaneously. 

\begin{proposition}
\label{thm:collab-2p-asymptotics1}
	Consider the two-player AC game of infinite horizon, where $H_0(a) = H_0(a')$. Let $s^\ast$ denote the strategy profile where each year the players invest their research potential into a single joint paper. Then the limit behavior for each player's utility under $s^\ast$ is
{$
	\limsup_{n \to \infty}\ {Util}^{s^\ast}_n \geq n/2.$}
\end{proposition}

\begin{proof}
If the claim holds for $h_0 = 0$, then it also holds for $h_0 > 0$, so assume that $h_0 = 0$. We use recursion to give a bound on $y^{s^\ast}_h$, the number of years needed to achieve an h-index of $h$ under $s^\ast$. We have that $y^{s^\ast}_0 = 0$, and $y^{s^\ast}_h \leq y^{s^\ast}_{\ceil{h/2} - 1} + h$, since after they have achieved h-index of $\ceil{h/2} - 1$, each of the following $h$ years they will produce a paper with at least $h$ citations each. We get the following bound:
\begin{align*}
y^{s^\ast}_h & \leq  y^{s^\ast}_{\ceil{h/2} - 1} + h 
 \leq  y^{s^\ast}_{\floor{h/2}} + h 
 \leq  h + h/2 + \ldots  \leq  2h.
\end{align*}
Conversely, $h \geq y^{s^\ast}_h/2$, so as the number of years $n$ goes to infinity, each player achieves a utility of
$\limsup\limits_{n \to \infty}\ {Util}^{s^\ast}_n \geq n/2.$
\end{proof}

\begin{proposition}
\label{thm:collab-2p-asymptotics2}
	Consider the two-player AC game of infinite horizon, where $H_0(a) = H_0(a')$. Let $s^\vDash$ denote the strategy profile where each year the players split their research potential evenly between two joint papers. Then the limit behavior for each player's utility under $s^\vDash$ is
{$\limsup_{n \to \infty}\ {Util}^{s^\ast}_n \sim 2\sqrt{n}.$}
\end{proposition}

\begin{proof}
If the claim holds for $h_0 = 0$, then it also holds for arbitrary initial citation profiles, so assume that $h_0 = 0$. Following strategy $s^\vDash$, from the time the players each reach an h-index of $h'$, it will take $\ceil{(h'+1)/2}$ years to accumulate $h'+1$ papers with $h'+1$ citations each. Thus a total of $n = \sum_{i=1}^h \ceil{i/2} \geq \frac{h(h+1)}{4}$ years are required to achieve an h-index of $h$. Conversely, as the number of years $n$ goes to infinity, each player achieves a utility of

\hspace*{0.2in}$
	\limsup\limits_{n \to \infty}\ {Util}^{s^\vDash}_n =
	\limsup\limits_{n \to \infty}\ \textstyle{\floor{\frac{-1 + \sqrt{1 + 16n}}{2}}} \sim 2\sqrt{n}.
$\hspace*{0.15in}
\end{proof}

We now examine how these two strategy profiles compare to other possible strategies for the two-player game.

\begin{lemma}
\label{thm:collab-2p-overtake1}
	Consider the two-player AC game of infinite horizon, where $H_0(a) = H_0(a')$. Let $s^\ast$ denote the strategy profile where each year the players invest their research potential into a single joint paper, and let $s^\vDash$ denote the strategy profile where each year the players split their research potential evenly between two joint papers. Let $S^{\{\ast,\,\vDash\}}$ denote the set of strategy profiles that each year prescribe either $s^\ast$ or $s^\vDash$. Then for any strategy profile $s \notin S^{\{\ast,\,\vDash\}}$, $\exists\ s' \in S^{\{\ast,\,\vDash\}}$ s.t.\ $h^{s'}_n$ overtakes both $h^s_n(a)$ and $h^s_n(a')$.
\end{lemma}

\begin{proof}
Consider a strategy profile $s \notin S^{\{\ast,\,\vDash\}}$. Consider the strategy profile $s'$ which is identical to $s$ for game states in which $s$ prescribes actions according to $s^\ast$ or $s^\vDash$, and behaves like $s^\ast$ otherwise. Let $y'$ be the first year in which $s$ and $s'$ differ, so that $H_{y'-1} = H^{s'}_{y'-1} = H^s_{y'-1}$. Let $P^s_{y'}$ denote the set of papers produced by $s$ in year $y'$, then we have $\sum_{p \in P^s_{y'}} cit(p) = 2(h_{y'-1} + 1)$. Since $s$ differs from $s^\vDash$ in year $y'$, there can be at most one paper with $\geq h_{y'-1} + 1$ citations; and since it differs from $s^\ast$, no paper can have $2(h_{y'-1} + 1)$ citations; it follows that $H^{s'}_{y'} \hpref H^s_{y'}(a)$ and $H^{s'}_{y'} \hpref H^s_{y'}(a')$. It follows by induction that $H^{s'}_y \hpref H^s_y(a)$ and $H^{s'}_y \hpref H^s_y(a')$ for all $y \geq y'$, and furthermore, $h^{s'}_y > h^s_y(a)$ and $h^{s'}_y > h^s_y(a')$ for all years $y \geq y'$ in which $h^{s'}_y$ increases. By definition, $h^{s'}_n$ overtakes $h^s_n(a)$ and $h^s_n(a')$.
\end{proof}

\begin{lemma}
\label{thm:collab-2p-overtake2}
	Consider the two-player AC game of infinite horizon, where $H_0(a) = H_0(a')$. Let $s^\ast$ denote the strategy profile where each year the players invest their research potential into a single joint paper. Then there does not exist a strategy profile $s \neq s^\ast$ such that either $h^s_n(a)$ or $h^s_n(a')$ overtakes $h^{s^\ast}_n$.
\end{lemma}

\begin{proof}
Consider a strategy profile $s \neq s^\ast$. Let $s^\vDash$ denote the
strategy profile where each year the players split their research
potential evenly between two joint papers, and let
$S^{\{\ast,\,\vDash\}}$ denote the set of strategy profiles that each
year prescribe either $s^\ast$ or $s^\vDash$. If $s \notin
S^{\{\ast,\,\vDash\}}$ then we are done by
Lemma~\ref{thm:collab-2p-overtake1}, so assume $s \in
S^{\{\ast,\,\vDash\}}$. Let $y^{s^\ast}_i$ denote the first year such
that $h^{s^\ast}_{y_i} \geq i$; let $y^s_i$ denote the first year such
that $h^s_{y_i} \geq i$; and let $k_i$ denote the number of years
$y^s_{i-1} \leq y < y^s_i$ in which $s$ differs from $s^\ast$. It
follows by induction that $y^{s^\ast}_i - y^s_i \leq k_i -
\sum\limits_{j < i} k_j$. In particular, if $k_i < \sum\limits_{j < i}
k_j$, then $y^{s^\ast}_i < y^s_i$, which implies that in year
$y^{s^\ast}_i$ we have $h^{s^\ast} > h^s$. Since the sequence $z_0 =
1,\ z_i = \sum\limits_{j < i} z_j$ grows exponentially yet $k_i$ can
grow at most linearly, this must happen an infinite number of
times. Since $h$ only takes integral values, we have that
$\liminf\limits_{n \to \infty} h^s_n - h^{s^\ast}_n < 0$, and so  $h^s_n$ does not overtake $h^{s^\ast}_n$.
\end{proof}

%

\begin{theorem}
\label{thm:collab-2p-equilibrium}
	Consider the two-player AC game of infinite horizon, where $H_0(a) = H_0(a')$. Let $s^\ast$ denote the strategy profile where each year the players invest their research potential into a single joint paper, and let $s^\vDash$ denote the strategy profile where each year the players split their research potential evenly between two joint papers. Let $S^{\{\ast,\,\vDash\}}$ denote the set of strategy profiles that each year prescribe either $s^\ast$ or $s^\vDash$. Then we have the following:
\begin{enumerate*}[(a)]
	\item \label{thm:collab-2p-equilibrium-all} All equilibria must be in $S^{\{\ast,\,\vDash\}}$.
	\item \label{thm:collab-2p-equilibrium-ast} The strategy profile $s^\ast$ is an equilibrium.
	\item \label{thm:collab-2p-equilibrium-not} Not all strategy profiles in $S^{\{\ast,\,\vDash\}}$ are equilibria.
\end{enumerate*}
\end{theorem}

\begin{proof}
Claims~(\ref{thm:collab-2p-equilibrium-all}) and~(\ref{thm:collab-2p-equilibrium-ast}) follow directly from Lemmas~\ref{thm:collab-2p-overtake1} and~\ref{thm:collab-2p-overtake2}, respectively.
For claim~(\ref{thm:collab-2p-equilibrium-not}), it is sufficient to show that $s^\vDash$ is not an equilibrium, which follows from Propositions~\ref{thm:collab-2p-asymptotics1} and~\ref{thm:collab-2p-asymptotics2} since the players would rather play according to $s^\ast$.
\end{proof}

The most striking consequence of this analysis is that working together, the authors can achieve {\em quadratically} more utility
than working alone. This only holds if they put all their effort into one joint paper; spreading their efforts across two (or more) joint papers is asymptotically no better than solo work.

\subsection{Multi-Player Game}
\label{sec:collab-multigame}

We now look at the AC game with an arbitrary number of players, $A$. For simplicity, we only analyze the case where $(\forall\ a \in A)\ H_0(a) = H_0$, i.e. initially all researchers have the same h-profile; the results can be generalized for arbitrary initial citation profiles.

We consider two variants: the ``static'' AC game, where each player follows the same collaboration strategy each year;
and the ``dynamic'' AC game, where new collaborations may be formed and the distribution of research potential may change.

We represent the static game as a directed graph, each edge $(a,a')$ labeled with a vector $\vec{\hat{q}_y^{a',a}}$ such that
\begin{itemize*}
	\item $(\forall\ a,a' \in A,\ i \in \mathbb{N}) \quad \vec{\hat{q}_y^{a,a'}}[i] \leq 1$; \quad and
	\item $(\forall\ a \in A) \quad \sum\limits_{i \in \mathbb{N}} \vec{\hat{q}_y^a}[i] + \sum\limits_{a' \neq a} \sum\limits_{i \in \mathbb{N}} \vec{\hat{q}_y^{a,a'}}[i] = 1$.
\end{itemize*}
That is, the graph dictates what fraction of a player's research potential is invested in each collaboration every year.


\begin{theorem}
\label{thm:collab-multi-equilibrium-static}
	Consider the static multi-player AC game of infinite horizon, where we have that $(\forall\ a \in A)\ H_0(a) = H_0$. Let $S^\ast$ be the set of strategy profiles corresponding to perfect matchings on $A$, where each year every pair of players in the matching invests their research potential into a single joint paper.\footnote{\upshape Note that this set is empty when $\abs{A}$ is odd.} Then all of the strategy profiles in $S^\ast$ are equilibria.
\end{theorem}

\begin{proof}

Consider a strategy profile $s^\ast \in S^\ast$.
It is obvious that no player can improve her utility if all other players' strategies remain the same, since joint papers are not possible without cooperation from both players.
Consider any strategy profile $s'$ differing from $s^\ast$ only in the strategies of players $a_1$ and $a_2$, so that under $s'$ both $a_1$ and $a_2$ invest a non-zero fraction of their research potential into a joint paper. By an argument similar to that in the proof of Lemma~\ref{thm:collab-2p-overtake2},
it is not possible that both $h^{s'}_n(a_1)$ overtakes $h^{s^\ast}_n(a_1)$ and $h^{s'}_n(a_2)$ overtakes $h^{s^\ast}_n(a_2)$, so by definition $a_1$ and $a_2$ do not form an unstable set. Since this is true for all pairs of players, there does not exist an unstable set of at most two players under $s^\ast$. Thus $s^\ast$ is an equilibrium.
\end{proof}

Next, we consider the same strategy profiles in the dynamic setting, with a very different result.

\begin{theorem}
\label{thm:collab-multi-equilibrium-dynamic}
	Consider the dynamic multi-player AC game of infinite horizon, where we have that $(\forall\ a \in A)\ H_0(a) = H_0$. Let $S^\ast$ be the set of strategy profiles corresponding to perfect matchings on $A$, where each year every pair of players in the matching invests their research potential into a single joint paper. Then for $\abs{A} > 2$, none of the strategy profiles in $S^\ast$ are equilibria.
\end{theorem}

\begin{proof}
Consider a strategy profile $s^\ast \in S^\ast$. Let $a_1$ and $a_2$ be two players who are not paired up in the matching, and let $a'_1$ and $a'_2$ be their matched pairs, respectively. We construct a strategy profile $s'$ as follows: All players besides $a_1$ and $a_2$ follow their respective strategies under $s^\ast$. In years 1 and 2, $a_1$ and $a_2$ follow their strategies under $s^\ast$; in years 3 and 7, they invest one unit of research potential in a joint paper with $a'_1$ and $a'_2$, respectively, and the rest in a single joint paper between themselves; and in all other years $a_1$ and $a_2$ invest all of their research potential in a single joint paper between themselves. It can be shown that $h^{s'}_n(a_1)$ overtakes $h^{s^\ast}_n(a_1)$ and $h^{s'}_n(a_2)$ overtakes $h^{s^\ast}_n(a_2)$. Therefore, $s^\ast$ is not an equilibrium.
\end{proof}

Although stable equilibria exist for the static multi-player AC game, the question of whether there exists an equilibrium for the dynamic version of the game is open. This discrepancy between the static and dynamic models suggests that we should be hesitant in drawing conclusions about real-world collaborative behavior from the static models assumed in the existing literature.

It may also be interesting to analyze the price of anarchy and stability of the AC game, which measure the extent to which researchers acting in their own self-interest benefit the academic community as a whole.

\section{Discussion}
\label{sec:discussion}

In order to facilitate analysis, we made several unrealistic modeling assumptions. For example, we assumed that all citations for a paper are received immediately upon publication. Taking a different citation model, such as constant additional citations per paper per year as in \cite{Hirsch@PNAS05}, a peak-decay model as in \cite{GR@JASIS09}, or that proposed in \cite{WSB@Science13} would complicate analysis but may be more plausible.

Our analysis only considered single- or two-authored papers. A more general model could allow multiple coauthors per paper, with conditions to prevent the degenerate solution of all researchers collaborating on a single giant paper. The model could also be extended to allow researchers to enter and retire from academia at different times.

There is also an inherent limitation in modeling human relationships and interactions. The underlying premise that people can be modeled as rational agents is itself subject to debate. Even if we take that to be a reasonable model, there are many more factors at play in the real world of academia -- e.g. geographic location, personal relationships, institutional loyalties, and academic competition -- than can be captured by a simple mathematical model.

\section*{Acknowledgments}

This work was sponsored by the NSF Grant 1101677: ICES: Auctions and Optimizations in Ad Exchanges.

\bibliographystyle{abbrv}
\bibliography{../hindex,../collab}

\end{document}